\documentclass[aps,pra,preprint,superscriptaddress,showpacs,nofootinbib]{revtex4-1}

\usepackage{amsmath}
\usepackage{amsfonts}
\usepackage{amsthm}
\usepackage{amssymb}
\usepackage{verbatim}
\usepackage{graphicx}
\usepackage{enumerate}
\usepackage{lineno}
\usepackage{color,xcolor}
\usepackage[T1]{fontenc}
\usepackage{mathptmx}
\usepackage{braket}
\usepackage{todonotes}
\usepackage{soul}
\usepackage{threeparttable}
\usepackage{multirow}
\usepackage{subfigure}
\usepackage{eucal}
\usepackage{ulem}
\graphicspath{{figures/}}

\usepackage[bookmarks=false,colorlinks,citecolor=blue,linkcolor=blue,anchorcolor=blue,urlcolor=blue]{hyperref}
\theoremstyle{remark}\newtheorem{theorem}{\hskip 1em Theorem}
\theoremstyle{remark}\newtheorem{lemma}{\hskip 1em Lemma}
\theoremstyle{remark}\newtheorem{example}{\hskip 1em Example}
\begin{document}

\title{Local unitary equivalence of arbitrary-dimensional multipartite quantum states}

\author{Qing Zhou}
\author{Yi-Zheng Zhen}
\author{Xin-Yu Xu}
\affiliation{Hefei National Research Center for Physical Sciences at the Microscale and School of Physical Sciences, University of Science and Technology of China, Hefei 230026, China}
\affiliation{CAS Center for Excellence in Quantum Information and Quantum Physics, University of Science and Technology of China, Hefei 230026, China}
\author{Shuai Zhao}
\affiliation{Hefei National Research Center for Physical Sciences at the Microscale and School of Physical Sciences, University of Science and Technology of China, Hefei 230026, China}
\affiliation{CAS Center for Excellence in Quantum Information and Quantum Physics, University of Science and Technology of China, Hefei 230026, China}
\affiliation{School of Cyberspace, Hangzhou Dianzi University, Hangzhou 310018, China}
\author{Wen-Li Yang}
\email{wlyang@nwu.edu.cn }
%\affiliation{Yukawa Institute for Theoretical Physics, Kyoto University, Kyoto 606-8502, Japan}
\affiliation{Institute of Modern Physics, Northwest University, Xi'an 710069, China}
\author{Shao-Ming Fei}
\email{feishm@cnu.edu.cn}
\affiliation{School of Mathematical Sciences, Capital Normal University, Beijing 100048, China}
%\affiliation{Max Planck Institute for Mathematics in the Sciences, 04103 Leipzig, Germany}
\author{Li Li}
\email{eidos@ustc.edu.cn}
\author{Nai-Le Liu}
\email{nlliu@ustc.edu.cn}
\author{Kai Chen}
\email{kaichen@ustc.edu.cn}
\affiliation{Hefei National Research Center for Physical Sciences at the Microscale and School of Physical Sciences, University of Science and Technology of China, Hefei 230026, China}
\affiliation{CAS Center for Excellence in Quantum Information and Quantum Physics, University of Science and Technology of China, Hefei 230026, China}
\affiliation{Hefei National Laboratory, University of Science and Technology of China, Hefei 230088, China}
\date{Feb 21, 2024}

\begin{abstract}
Local unitary equivalence is an important ingredient for quantifying and classifying entanglement. 
%In this work, we address the issue of verifying whether two quantum states are local unitary equivalent or not, which has been answered in the case of multipartite pure states. However, for mixed states, 
Verifying whether or not two quantum states are local unitary equivalent is a crucial problem, where only the case of multipartite pure states is solved. For mixed states, however, the verification of local unitary equivalence is still a challenging problem. 
%especially for the case that degenerate density matrices. 
In this paper, based on the coefficient matrices of generalized Bloch representations of quantum states, we find a variety of local unitary invariants for arbitrary-dimensional bipartite quantum states. These invariants are operational and can be used as necessary conditions for verifying the local unitary equivalence of two quantum states. Furthermore, we extend the construction to the arbitrary-dimensional multipartite case. %Compared with the existing results, our approach presents the invariants under local unitary transformations that can verify. 
% and effective strategy to construct necessary conditions for multipartite case. 
We finally apply these invariants to estimate concurrence, a vital entanglement measure, showing the practicability of local unitary invariants in characterizing entanglement.
\end{abstract}

%\pacs{03.67.Dd, 03.67.Hk, 03.67.Mn, 03.65.Ud}

\maketitle

\section{Introduction}\label{sec:introduction}
Entanglement as an important quantum resource plays a crucial role in quantum information technologies \cite{nielsen2010quantum}, such as quantum teleportation \cite{bennett1993teleporting,bouwmeester1997experimental}, quantum secret sharing \cite{RevModPhys.81.865}, quantum cryptography \cite{ekert1991quantum} and quantum repeaters \cite{briegel1998quantum}. However, the quantification and classification of entanglement remain a challenging problem. The local unitary (LU) equivalence of the quantum state is an effective ingredient to characterize entanglement, the amount of which is invariant under LU transformations \cite{dur2000three}. The LU invariants of a quantum state can thus be adopted to detect entanglement or entanglement classes and to derive, e.g., entanglement measures \cite{PhysRevLett.95.040504,PhysRevLett.95.210501} and separability criteria \cite{CHEN200214,10.5555/2011534.2011535,PhysRevA.67.032312,PhysRevA.68.062313}. If one can completely solve the problem of the LU equivalence class, there will be a major breakthrough in understanding and making good use of entanglement resource. %Therefore, it is crucial to develop criteria for verifying LU equivalence of two arbitrary quantum states.

%Up to now, entanglement detection remains an intractable problem. As the entanglement properties of entangled states are conserved via local unitary(LU) operators \cite{dur2000three} and LU equivalence is an important ingredient for quantifying and classifying entanglement, we can study entanglement by classifying quantum states. Therefore, it is meaningful to study the criteria for LU equivalence of any two quantum states.

In recent years, much effort has been devoted to characterizing and understanding the LU equivalence of quantum states \cite{albeverio2003nonlocal,albeverio2005equivalence, sun2005note, zhang2013criterion, li2014local, albeverio2005multipartite, albeverio2005local, albeverio2007local, zhang2013criterion, martins2015necessary, linden1999nonlocal, albeverio2005equivalence, sun2017local, cui2017local, jing2015local,Li_2013,Jing_2014,Zhang_2015}. If two quantum states are LU equivalent, there must be some invariants under LU transformation. Thus, the basic method to determine whether or not two quantum states are LU equivalent is by checking the LU invariants. Up to now, many novel methods for computing the LU invariants have been developed \cite{PhysRevA.58.1833,817508}. For pure states, an elegant result of the LU equivalence problem for multipartite qubit states has been obtained by Kraus \cite{PhysRevA.82.032121,kraus2010local}. The problem for arbitrary-dimensional pure states case has also been solved in Ref. \cite{liu2012local}. 
However, for mixed cases, only two-qubit states and some special cases have been solved.
For instance, a set of 18 polynomial invariants is proposed to completely determine whether a two-qubit mixed quantum state is LU equivalent to another state \cite{makhlin2002nonlocal}.
As for arbitrary dimensional bipartite quantum states, only a set of LU invariants for nondegenerate density matrices is known \cite{zhou2012local}.
In general, it is extremely difficult to present a sufficient and necessary condition for certifying the LU equivalence of two arbitrary quantum states.

In this paper, we focus on building LU invariants, which are the necessary condition for the LU equivalence for arbitrary-dimensional multipartite mixed quantum states.\ The previous results for LU problems are to calculate LU invariants from the density matrices directly, which are inconvenient to deal with the degeneracy problem. Here, we derive LU invariants based on the singular value decomposition (SVD) of coefficient matrices under the generalized Bloch representations of quantum states. In this manner, the difficulty due to degeneration of the density matrix can be avoided. Moreover, our method for deriving LU invariants of bipartite states can be conveniently extended to arbitrary dimensional multipartite cases. As an example, we show its application in the case of three-qudit states. Furthermore, the concurrence as an entanglement measure \cite{PhysRevA.64.042315,PhysRevLett.95.260502,PhysRevLett.92.167902,PhysRevA.62.032307,PhysRevLett.80.2245,con2002} can be derived from our LU invariants. A quadratic relation is presented here between the LU invariants and the concurrence, which reveals the significance and practicability of LU invariants in exploring entanglement.

%In this paper, we shall firstly develop a set of LU invariants for arbitrary dimensional bipartite quantum states. The new LU invariants are derived based on the singular value decomposition (SVD) characterization of quantum states in the generalized Bloch representation. In this manner, the difficulty due to degeneration of density matrix can be avoided. Then, our method can be conveniently extended to arbitrary dimensional multipartite cases. As an example, we show its application in the case of three-qudit states. Moreover, an approximate quadratic relation has been presented between the invariants and the concurrence.

The paper is organized as follows. 
In Sec. \ref{sec:twoqudit}, we derive necessary conditions for the LU equivalence of any two bipartite mixed quantum states, from generalized Bloch representations of the density matrices. 
%We also show its performance in the case of degenerate density matrices.
In Sec. \ref{sec:multi}, we extend our methods to the multipartite case and study the LU equivalence for three-qudit states. In Sec. \ref{sec:entanglement}, by analytical calculations, the relation between the LU invariants and the concurrence is explored here for pure bipartite states. Finally, we conclude with our results in Sec. \ref{sec:conclusion}.

\section{Two-qudit case}\label{sec:twoqudit}
In this section, based on the generalized Bloch representations of density matrices of arbitrary dimensional bipartite quantum states \cite{eltschka2020maximum,appel2020monogamy}, we present the necessary conditions for LU equivalence in matrix form and then try to find the operational invariants under LU transformations.

Let ${\cal H}_{1}$ and $ {\cal H}_{2}$ be two complex Hilbert spaces with ${\rm dim}({\cal H}_{1})={\rm dim}({\cal H}_{2})=N$, and $\rho$ be a density matrix of a bipartite quantum mixed state acting on ${\cal H}_{1}\otimes {\cal H}_{2}$. Without loss of generality, $\rho$ can be expressed as:
\begin{equation}
\rho=\frac{1}{N^{2}}I\otimes I+\sum_{i}R_{i}\lambda_{i}^{1}\otimes I+\sum_{j}S_{j}I\otimes \lambda_{j}^{2}+\sum_{ij}T_{ij}\lambda_{i}^{1}\otimes\lambda_{j}^{2},
\label{eq:rho}
\end{equation}
where $T_{ij}=\frac{1}{4}Tr(\rho\lambda_{i}^{1}\otimes\lambda_{j}^{2})$, $S_{j}=\frac{1}{2N}Tr(\rho I\otimes\lambda_{j}^{2})$, $R_{i}=\frac{1}{2N}Tr(\rho\lambda_{i}^{1}\otimes I)$, $i,j=1,2,...,N^{2}-1$. Here, $\lambda_{i}$ are the generators of $SU(N)$ with $Tr(\lambda_{i}\lambda_{j})=2\delta_{ij}$. The $T_{ij},S_{j},R_{i}$ are all real coefficients, as the state $\rho$ and $SU(N)$'s generators are all Hermitian.

%$S=(S_{j})$ and $R=(R_{i})$ are all $N^{2}-1$ dimensional column vectors, $T=(T_{ij})$ is a $(N^{2}-1)\times(N^{2}-1)$ matrix,   $T_{ij}=\frac{1}{N^{2}}Tr(\rho\lambda_{i}^{1}\otimes\lambda_{j}^{2})$, $S_{j}=\frac{1}{N^{2}}Tr(\rho I\otimes\lambda_{j}^{2})$, $R_{i}=\frac{1}{N^{2}}Tr(\rho\lambda_{i}^{1}\otimes I)$. Because of the Hermiticity of density matrix and $SU(N)$'s generators, it is easy to know that all the coefficients are real. 

The two-qudit states $\rho$ and another two-qudit state $\rho^{\prime}=\frac{1}{N^{2}}I\otimes I+\sum_{i}R_{i}^{\prime}\lambda_{i}^{1}\otimes I+\sum_{j}S_{j}^{\prime}I\otimes \lambda_{j}^{2}+\sum_{ij}T_{ij}^{\prime}\lambda_{i}^{1}\otimes\lambda_{j}^{2}$ 
 %\begin{equation}
%\rho^{\prime}=\frac{1}{N^{2}}I\otimes I+\sum_{i}R_{i}^{\prime}\lambda_{i}^{1}\otimes I+\sum_{j}S_{j}^{\prime}I\otimes \lambda_{j}^{2}+\sum_{ij}T_{ij}^{\prime}\lambda_{i}^{1}\otimes\lambda_{j}^{2}.
%\end{equation}
are called LU equivalent, if there exist unitary operators $U_{1},U_{2}\in SU(N)$, such that %$\rho^{\prime}=(U_{1}\otimes U_{2} )\rho(U_{1}\otimes U_{2})^{\dag}$,
\begin{equation}
\rho^{\prime}=(U_{1}\otimes U_{2} )\rho(U_{1}\otimes U_{2})^{\dag}.
\label{eq:LU}
\end{equation}
%where $\dag$ stands for conjugate transposition. 
Since the coefficients $\{R_{i},S_{j},T_{ij}\}_{i,j}$ capture complete information of $\rho$, one can investigate the LU equivalence based on the ``feature matrix'' $M(\rho)$ defined as $M(\rho)={\left[ \begin{array}{cc}
1 & S^{t}\\
R & T
\end{array}
\right]}$. Here, the $S=(S_{j})$ and $R=(R_{i})$ are $N^{2}-1$ dimensional column vectors and $T=(T_{ij})$ is a $(N^{2}-1)\times(N^{2}-1)$ matrix. Similarly, one can define $M(\rho^{\prime})={\left[ \begin{array}{cc}
1 & S^{\prime t}\\
R^{\prime} & T^{\prime}
\end{array}
\right]}$. %Then the following Lemma can be proved.% \cite{cui2017local}.

It is shown in Ref.$~$\cite{cui2017local} that if two bipartite mixed states $\rho$ and $\rho^{\prime}$ are LU equivalent, then there are $O_{1},O_{2}\in SO(N^{2}-1)$ such that $R^{\prime}=O_{1}^{t}R,S^{\prime}=O_{2}^{t}S,T^{\prime}=O_{1}^{t}TO_{2}$. Then the following Lemma can be presented.

\begin{lemma}
If two bipartite mixed states $\rho$ and $\rho^{\prime}$ are LU equivalent, there are $ O_{1},O_{2}\in SO(N^{2}-1)$, such that $M(\rho^{\prime})={\rm diag}\{1,O_{1}^{t}\}M(\rho){\rm diag}\{1,O_{2}\}$, where $t$ stands for transposition.
\end{lemma}
% \begin{proof}
% Based on the result in Ref.$~$\cite{cui2017local}, if two bipartite mixed states $\rho$ and $\rho^{\prime}$ are LU equivalent, there are $O_{1},O_{2}\in SO(N^{2}-1)$, such that $R^{\prime}=O_{1}^{t}R,S^{\prime}=O_{2}^{t}S,T^{\prime}=O_{1}^{t}TO_{2}$. Hence, $M(\rho^{\prime})=diag\{1,O_{1}^{t}\}M(\rho)diag\{1,O_{2}\}$ is a direct result from above relations.
% \end{proof}
From Lemma 1 we obtain the transformation of the quantum state's feature matrix under LU transformations. As the above Lemma is based on the generalized Bloch representations, it will be restricted by the properties of Bloch-vector space \cite{kimura2003bloch}. 
%Some generalized Bloch representations may be not a density matrix when $N>2$ \cite{kimura2003bloch}. 
Hence, the special orthogonal matrices $O\in SO(N^{2}-1)$ acting on the coefficient vectors of $\rho$ may transform $\rho$ into an unphysical form that does not satisfy the definition of a density matrix. 
% which does not satisfy the definition of density matrix. 
%Namely, $SU(N)$ and $SO(N^{2}-1)$ is not double-covering when $N>2$. 
Hence, the Lemma 1 is a necessary but not sufficient condition for the LU equivalence of quantum states when $N>2$. 
%Meanwhile, it is not operational to verify LU equivalence directly, as verifying the existence of these special orthogonal matrices is a hard mathematical problem. Hence, for simplicity, finding invariants under LU transformations that can be easily calculated is necessary. %For this reason, we will decompose $T$ to reduce the difficulty of finding invariants under LU transformations.

%\section{Two-qudit case}\label{sec:twoqudit}
To make the conditions for verifying LU equivalence operational, we will find further the LU invariants based on Lemma 1. First, let us focus on matrix $T$ and make a SVD of it. As $T$ is a real matrix, one can obtain a real orthogonal matrix $P$ and $Q$ with respect to the corresponding singular values, such that $T=P\Sigma Q^{t},$ 
%\begin{equation}
%T=P\Sigma Q^{t},
%\end{equation}
where $\Sigma={\rm diag}\{\mu_{1},\mu_{2},...,\mu_{d},0,...,0\}.$ Here, the singular values are in decreasing order, and $d$ is the $rank$ of matrix $T$. To deal with the degenerate singular values, one can divide $\Sigma$ into a block matrix by putting the same singular values in a submatrix, such that $\Sigma={\rm diag}\{\Sigma_{1},\Sigma_{2},...,\Sigma_{n},\Sigma_{n+1}\}$, where $n+1$ means that there are $n+1$ different singular values. The $\Sigma_{m}$ is an $a_{m}\times a_{m}$ matrix %$m\in\{1,2,\dots,n+1\}$ 
%and the $a_{m}$ is the number of the same singular values and 
with $a_{1}+a_{2}+\dots+a_{n+1}=N^{2}-1$. For $m\in\{1,2,\dots,n\}$, there is $\Sigma_{m}\propto I$. When it comes to the case of $m=n+1$, one should consider whether there exists a zero singular value of matrix $T$. If $T$ is a full rank matrix, the $\Sigma_{n+1}$ is also proportional to the identity matrix, otherwise $\Sigma_{n+1}$ is a zero matrix. 
%If $rank(T)=N^{2}-1$, the $\Sigma_{m}$ is proportion to identity matrix for any $m\in\{1,\cdots,n+1\}$; otherwise $\Sigma_{m}\propto I$ for $m\in\{1,2,\dots,n\}$ with a zero matrix $\Sigma_{n+1}$. 
Thus, the matrix $M$ can be rewritten as $M(\rho)={\left[ \begin{array}{cc}
1 & \\
 & P
\end{array}
\right]}{\left[ \begin{array}{cc}
1 & \widetilde{S}^{t}\\
\widetilde{R} & \Sigma
\end{array}
\right]} {\left[ \begin{array}{cc}
1 & \\
 & Q^{t}
\end{array}
\right]},$
%\begin{equation}
%M(\rho)={\left[ \begin{array}{cc}
%1 & \\
% & P
%\end{array}
%\right]}{\left[ \begin{array}{cc}
%1 & \widetilde{S}^{t}\\
%\widetilde{R} & \Sigma
%\end{array}
%\right]} {\left[ \begin{array}{cc}
%1 & \\
% & Q^{t}
%\end{array}
%\right]},
%\end{equation}
where $\widetilde{R}\equiv P^{t}R,(\widetilde{S})^{t}\equiv S^{t}Q$. Let $\pi_{m}(R)$ denote the projection of the $m$-th part of vector $R$ and and $\|*\|$ be the Euclid norm \cite{horn1990matrix}. Here, each part $\pi_{m}(R)$ is an $a_{m}$ dimensional vector, $m=1,2,\cdots,n+1$. 
%and define $\|*\|$ as Euclid norm \cite{horn1990matrix}. 
With these notations, we present our main result as the following theorem.
\begin{theorem}
If two two-qudit states $\rho$ and $\rho^{\prime}$ are LU equivalent, one has the following invariants: %$(1)\Sigma, det(T), det(M(\rho))\ (2)\|\pi_{m}\widetilde{R}\|,\|\pi_{m}\widetilde{S}\|,m=1,2,...,n+1,\ (3)(\pi_{m}\widetilde{S})^{t}(\pi_{m}\widetilde{R}),m=1,2,...,n^{\prime},$
\begin{equation}
\begin{split}
&(1)\Sigma,{\rm det}(T), {\rm det}(M(\rho)),\\&(2)\|\pi_{m}(\widetilde{R})\|,\|\pi_{m}(\widetilde{S})\|,m=1,2,...,n+1,\\&(3)\pi_{m}(\widetilde{S})^{t}\pi_{m}(\widetilde{R}),m=1,2,...,n^{\prime},
\label{eq:3}
\end{split}
\end{equation}
where $n^{\prime}$ is determined by matrix $T$. If $T$ has no zero singular values, one has $n^{\prime}=n+1$, otherwise, $n^{\prime}=n$.
\end{theorem}
%The detailed proof can be found in Appendix.
\begin{proof}
From the Lemma 1, it is easy to get ${\rm det}(T^{\prime})={\rm det}(T)$ and ${\rm det}(M(\rho^{\prime}))={\rm det}(M(\rho))$, as ${\rm det}(O_{1}^{t})={\rm det}(O_{2})={\rm det}({\rm diag}\{1,O_{1}^{t}\})={\rm det}({\rm diag}\{1,O_{2}\})=1$ with $O_{1},O_{2}\in SO(N^{2}-1)$. Similarly, one can perform a SVD of matrices $T^{\prime}$ with $T^{\prime}=P^{\prime}\Sigma^{\prime} Q^{\prime t}$. 
If $\rho$ and $\rho^{\prime}$ are LU equivalent, from Lemma 1 one has $T^{\prime}=O_{1}^{t}P\Sigma Q^{t}O_{2}$. 
As $O_{1}^{t}P$ and $Q^{t}O_{2}$ are orthogonal matrices, one has $\Sigma^{\prime}=\Sigma$. To prove other invariants, the degeneracy of matrix $T$ should be considered. If matrix $T$ is degenerate, $P$ and $Q$ are not unique. 
%One should consider the transformation of $P$ and $Q$ in general case. For simplicity, the non-uniqueness can be considered as an extra $\Sigma$-preserving operation performed on certain SVD of $T$. 
Firstly, from Lemma 1, one can fix $P^{\prime},Q^{\prime}$ as $\ P^{\prime}=O_1^{t}P,\ Q^{\prime}=O_{2}^{t}Q.$ 
By these relations, one has
\begin{equation}
\widetilde{R}^{\prime}=P^{\prime t}R^{\prime}=P^{t}O_{1}O_{1}^{t}R=\widetilde{R},\ \widetilde{S}^{\prime t}=S^{\prime t}Q^{\prime}=\widetilde{S}^{t}.
\label{eq:relation-RS}
\end{equation}
Next, let us discuss the general case 
%that $P$ and $Q$ are not fixed 
when $T$ has zero singular values. 
For $m=1,2,...,n$, as $\Sigma_{m}\propto I$, the $\Sigma_{m}$-preserving operations are orthogonal matrices $G_{m}$ with $G_{m}\Sigma_{m}G_{m}^{t}=\Sigma_{m}$. For $m=n+1$, the orthogonal matrices $G_{n+1}$, 
%$O_{n+1}^{\bullet}$ 
$\mathcal{G}_{n+1}$ are $\Sigma_{m}$-preserving operations due to $G_{n+1}\Sigma_{n+1}\mathcal{G}^{t}_{n+1}=\Sigma_{n+1}$. Here, there is no need for the orthogonal matrices $G_{n+1}$ and $\mathcal{G}_{n+1}$ to be the same, as $\Sigma_{n+1}$ is the zero matrix. Thus, one can construct $\widetilde{O}={\rm diag}\{G_{1},G_{2},...,G_{n},G_{n+1}\}$ and $\widetilde{\mathcal{O}}={\rm diag}\{G_{1},G_{2},...,G_{n},\mathcal{G}_{n+1}\}$, such that $\widetilde{O}\Sigma\widetilde{\mathcal{O}}^{t}=\Sigma$, where $\widetilde{O},\ \widetilde{\mathcal{O}}$ are orthogonal matrices.
For a general SVD of matrix $T$, the feature matrix can be written as 
\begin{equation}
M(\rho)={\left[ \begin{array}{cc}
1 & \\
 & P\widetilde{O}
\end{array}
\right]}{\left[ \begin{array}{cc}
1 & \widetilde{S}^{\S t}\\
\widetilde{R}^{\S} & \Sigma
\end{array}
\right]} {\left[ \begin{array}{cc}
1 & \\
 & \widetilde{\mathcal{O}}^{t}Q^{t}
\end{array}
\right]},
\end{equation}
where $\widetilde{R}^{\S}\equiv \widetilde{O}^{t}\widetilde{R},\ \widetilde{S}^{\S t}\equiv \widetilde{S}^{t}\widetilde{\mathcal{O}}$. Similarly, for state $\rho^{\prime}$, one has $\widetilde{R}^{\prime\S}\equiv \widetilde{O}^{\prime t}\widetilde{R}^{\prime},\ \widetilde{S}^{\prime\S t}\equiv \widetilde{S}^{\prime t}\widetilde{\mathcal{O}^{\prime}}$, with orthogonal matrices $\widetilde{O^{\prime}}={\rm diag}\{G^{\prime}_{1},G^{\prime}_{2},...,G^{\prime}_{n},G^{\prime}_{n+1}\}$ and $\widetilde{\mathcal{O}^{\prime}}={\rm diag}\{G^{\prime}_{1},G^{\prime}_{2},...,G^{\prime}_{n},\mathcal{G}^{\prime}_{n+1}\}$. 
From Eq. (\ref{eq:relation-RS}), one has $\pi_{m}(\widetilde{R}^{\prime})=\pi_{m}(\widetilde{R}),\
\pi_{m}(\widetilde{S}^{\prime})=\pi_{m}(\widetilde{S}).$ 
Then, from the definition of $\widetilde{R}^{\S},\widetilde{S}^{\S}$, 
it is easy to check that $\|\pi_{m}(\widetilde{R}^{\S})\|= \|\pi_{m}(\widetilde{R}^{\prime\S})\|,\|\pi_{m}(\widetilde{S}^{\S})\|=\|\pi_{m}(\widetilde{S}^{\prime\S})\|$, 
where $m=1,2,\dots,n+1$. Furthermore, for $m=1,2,\dots,n$, one can also obtain $\pi_{m}(\widetilde{S}^{\S})^{t}\pi_{m}(\widetilde{R}^{\S})=\pi_{m}(\widetilde{S}^{\prime\S})^{t}\pi_{m}(\widetilde{R}^{\prime\S})$ by direct calculation. 
Moreover, if $T$ has no zero singular values, the $\Sigma_{n+1}$ is not a zero matrix. In this case, one has $\mathcal{G}_{n+1}=G_{n+1}$ and then $\pi_{n+1}(\widetilde{S}^{\S})^{t}\pi_{n+1}(\widetilde{R}^{\S})=\pi_{n+1}(\widetilde{S}^{\prime\S})^{t}\pi_{n+1}(\widetilde{R}^{\prime\S})$. 
Without loss of generality, we can erase the notation $\S$ and summarize the results in Theorem 1.
\end{proof} 
%From Theorem 1, if we have two two-qudit states $\rho$ and $\rho^{\prime}$, we firstly present the generalized Bloch representations of them. The feature matrices will be obtained. Then, by executing the SVD of coefficient matrix $T=P\Sigma Q^{t}$, the transformed coefficient vectors $\widetilde{R}=P^{t}R$ and $\widetilde{S}=Q^{t}S$ are obtained. Meanwhile, we perform the same operations for $\rho^{\prime}$. Finally, we verify the invariants in Theorem 1. Therefore, we can easily verify the LU equivalence of any two two-qudit states regardless of whether density matrices are degenerate or not by Theorem 1 necessarily. 
From the Theorem 1, one can verify that two two-qudit states are not LU equivalent, if the invariants shown in the Eq.$~$(\ref{eq:3}) are not equal for two states. Compared with the existing results, the invariants in the Theorem 1 covers the one of Ref. \cite{cui2017local} and have advantages in verifying the LU equivalence of high-dimensional mixed states. %and are constructed from the novel point of feature matrix. %Moreover, because of the extra invariant $det(M(\rho))$, our result has an advantage in verifying LU equivalence over existing result in Ref. \cite{cui2017local}. 

\begin{example}
By introducing certain noise into a nonmaximally entangled two-qutrit state $|\psi\rangle=\frac{\sqrt{2}}{4}|00\rangle+\frac{\sqrt{2}}{4}|11\rangle+\frac{\sqrt{3}}{2}|22\rangle$, one has ${\rho}=q|\psi\rangle\langle\psi|+\frac{1-q}{6}(|01\rangle\langle 01|+|10\rangle\langle 10|+|02\rangle\langle 02|+|20\rangle\langle 20|+|12\rangle\langle 12|+|21\rangle\langle 21|)$. Assuming $q=\frac{4}{17}$, the coefficients of Bloch representation can be obtained as $R=S=(0,0,0,0,0,0,0,-5\sqrt{3}/306)$ and 
\begin{equation}\nonumber
T=
\begin{bmatrix}
\frac{1}{68} & 0&0  &0 &0 &0 & 0&0 \\
 0&\frac{\sqrt{6}}{68}& 0& 0&0 &0 & 0&0\\
 0& 0 & -\frac{1}{68} &0 &0 &0 &0 &0 \\
 0& 0& 0&-\frac{5}{102}& 0&0 &0 & 0\\
  0&0 & 0& 0&\frac{\sqrt{6}}{68}& 0& 0&0 \\
    0&0 &0 &0 & 0& -\frac{\sqrt{6}}{68}&0 & 0\\
       0 &0 & 0&0 &0 &0 & -\frac{\sqrt{6}}{68}& 0 \\
           0 & 0& 0&0 &0 & 0& 0& 0\\
\end{bmatrix}.
%T=diag\{1/68,\sqrt{6}/68,-1/68,-5/102,\sqrt{6}/68,-\sqrt{6}/68,-\sqrt{6}/68,0\}.
\end{equation}
Now, one can construct another states ${\rho}^{\prime}$ with coefficients $R^{\prime}=R,\ S^{\prime}=S,\ T^{\prime}=-T$. 
\end{example}
In this case, as the LU invariants in Ref.$~$\cite{cui2017local} for the above two states are all the same, one cannot verify whether or not the two states $\rho$ and $\rho^{\prime}$ are LU equivalent. However, by our method, one can check that 
${\rm det}(M(\rho))\neq {\rm det}(M(\rho^{\prime}))$, 
%$det(M(\rho))=-det(M(\rho^{\prime}))=\frac{1}{(256)^{2}}$, 
which means $\rho$ and $\rho^{\prime}$ have different invariants. Hence, $\rho$ and $\rho^{\prime}$ are not LU equivalent. Our results cover the invariants in Ref.$~$\cite{cui2017local} and we have an extra invariant ${\rm det}(M(\rho))$. The invariants derived in Ref. \cite{cui2017local} are based on the orbit of group, where it cannot naturally lead to the discovery of an invariant such as ${\rm det}(M(\rho))$. However, our construction of invariants is based on a feature matrix, which binds together the coefficient vectors $R$, $S$ and coefficient matrix $T$. Moreover, by performing a SVD of matrix $T$, we explore the fine-grained relation between vectors $R$ and $S$ according to singular values of $T$. In most of the existing works, they study the coefficient vectors and coefficient matrix in isolation or bind them by the orbit of the group. The construction of the feature matrix brings different points of view to explore the LU invariants, which can give a more refined characterization of LU equivalence.

%In this case, the LU equivalence can be verified by our LU invariants, however, cannot be checked by previous works$~$\cite{cui2017local}. As $det(M(\bar{\rho}))\neq det(M(\bar{\rho}^{\prime}))$, one can check that the above two states are not LU equivalent. 

By the above discussion, our criteria for verifying the LU equivalence of arbitrary dimensional mixed states are superior, as they can be applied to more scenarios.

%Compare to the results in Ref.\cite{cui2017local}, our invariants in theorem 1 can recover their results. Further, our results are more intuitive from Bloch representation.
\section{three-qudit case}\label{sec:multi}
To show that the method can be applied to a multipartite case, we firstly consider the three-qudit system. Then, with a simple extension, the generalization to the general multipartite case is immediate.
% The method can also be applied to multipartite case. 
% To present the general results, we just consider the three-qudit system, with simple extension to general multipartite case being immediate. 
%and then the multipartite case is similar. 
The generalized Bloch representation of a density matrix for a three-qudit state can be written as 
%\begin{equation}
%\begin{split}
%\rho=\frac{1}{N^{3}}I\otimes I\otimes I+\sum_{i}T_{1}^{i}\lambda_{i}^{1}\otimes I\otimes I+\sum_{j}T_{2}^{j}I\otimes\lambda_{j}^{2}\otimes I+\sum_{k}T_{3}^{k}I\otimes I\otimes \lambda_{k}^{3}+\sum_{ij}T_{ij}\lambda_{i}^{1}\otimes\lambda_{j}^{2}.
%\label{eq:rho3}
%\begin{split}
%\end{equation}
\begin{equation}
\begin{split}
\rho=&\frac{1}{N^{3}}I^{1}I^{2}I^{3}+\sum_{i}T_{i}^{1}\lambda_{i}^{1}I^{2}I^{3}+\sum_{j}T_{j}^{2}I^{1}\lambda_{j}^{2}I^{3}\\&+\sum_{k}T_{k}^{3}I^{1}I^{2}\lambda_{k}^{3}+\sum_{ij}T_{ij}^{12}\lambda_{i}^{1}\lambda_{j}^{2}I^{3}+\sum_{ik}T_{ik}^{13}\lambda_{i}^{1}I^{2}\lambda_{k}^{3}\\&+\sum_{jk}T_{jk}^{23} I^{1}\lambda_{j}^{2}\lambda_{k}^{3}+\sum_{ijk}T_{ijk}^{123} \lambda_{i}^{1}\lambda_{j}^{2}\lambda_{k}^{3},
\label{eq:rho3}
\end{split}
\end{equation}
where $\lambda_{i}^{1},\lambda_{j}^{2},\lambda_{k}^{3}$ are the generators of $SU(N)$ with $i,j,k=1,2,\dots,N^{2}-1$. Here, the $T_{1}=\{T_{i}^{1}\},\ T_{2},\ T_{3}$ are three $N^{2}-1$ dimensional coefficient vectors, the $T_{12}=\{T_{ij}^{12}\},\ T_{13},\ T_{23}$ are three $(N^{2}-1)\times(N^{2}-1)$ dimensional coefficient matrices and $T_{123}=\{T_{ijk}^{123}\}$ is a coefficient tensor. Here, one can write $I\otimes I\otimes I$ as $I^{1}I^{2}I^{3}$ by omitting the notation $\otimes$. Some results in the three-qudit system are similar to the two-qudit case. If two mixed states $\rho$ and $\rho^{\prime}$ are LU equivalent, there are some $O_{1},O_{2},O_{3}\in SO(N^{2}-1)$, such that
\begin{equation}
    \begin{aligned}
        {T}_{1}^{\prime}&={O}_{1}^{t}{T}_{1},\ {T}_{2}^{\prime}={O}_{2}^{t}{T}_{2},\ {T}_{3}^{\prime}={O}_{3}^{t}{T}_{3},\\
        {T}_{12}^{\prime}&={O}_{1}^{t}{T}_{12}{O}_{2},\ {T}_{13}^{\prime}={O}_{1}^{t}{T}_{13}{O}_{3},\ {T}_{23}^{\prime}={O}_{2}^{t}{T}_{23}{O}_{3},\\
            {T}_{123}^{\prime}&=( {O}_{1}^{t}\otimes    {O}_{2}^{t}\otimes    {O}_{3}^{t})  {T}_{123}.
    \end{aligned}   \label{eq:3qudit}
\end{equation}
% \begin{equation}
% T_{\alpha}^{\prime}=O_{\alpha}^{t}T_{\alpha},T_{\alpha\beta}^{\prime}=O_{\alpha}^{t}T_{\alpha\beta}O_{\beta},T_{123}^{\prime}=(O_{1}^{t}\otimes O_{2}^{t}\otimes O_{3}^{t})T_{123}.
% \label{eq:3qudit}
% \end{equation}
Now, let the elements in $T_{123}$ be arranged as $(N^{2}-1)\times(N^{2}-1)^{2}$ matrices $T_{1|23},T_{2|13}$, and $T_{3|12}$, whose details are presented in the Appendix. 
%with the $\alpha^{\prime}$-th subscript of element $T_{ijk}^{123}$ as the row index of the matrix $T_{\alpha^{\prime}|\beta^{\prime}l}$, where $\alpha^{\prime},\beta^{\prime},l=1,2,3$ and $\alpha^{\prime}\neq \beta^{\prime},\alpha^{\prime}\neq l,\beta^{\prime}<l$. 
For example, 
\begin{equation}
T_{1|23}=
\begin{bmatrix}
T_{111} & T_{112}& \cdots & T_{1(N^{2}-1)(N^{2}-1)}\\
T_{211} &T_{212}& \cdots & T_{2(N^{2}-1)(N^{2}-1)}\\
\vdots & \vdots & \ddots & \vdots\\
T_{(N^{2}-1)11} & T_{(N^{2}-1)12}& \cdots & T_{(N^{2}-1)(N^{2}-1)(N^{2}-1)}
\end{bmatrix},
\end{equation}
%the first column of $T_{1|23}$ is $\{T_{111},T_{211},T_{311}\cdots, T_{(N^{2}-1)11}\}$, 
where we let the first subscript be the row index. From Eq. (\ref{eq:3qudit}), one has 
\begin{equation}
    {T}_{\alpha|\beta\gamma}^{\prime}={O}_{\alpha}^{t}{T}_{\alpha|\beta\gamma}{O}_{\beta}\otimes {O}_{\gamma},\ {T}_{\beta\gamma}^{\prime}={O}_{\beta}^{t} {T}_{\beta\gamma}{O}_{\gamma},
    \label{eq:1|23}
\end{equation}
% \begin{equation}
% %\begin{split}
% T_{\alpha^{\prime}|\beta^{\prime} l}^{\prime}=O_{\alpha^{\prime}}^{t}T_{\alpha^{\prime}|\beta^{\prime}l}O_{\beta^{\prime}}\otimes O_{l},T_{\beta^{\prime} l}^{\prime}=O_{\beta^{\prime}}^{t} T_{\beta^{\prime} l}O_{l},
% \label{eq:1|23}
% %\end{split}
% \end{equation}
where $(\alpha,\beta,\gamma)=(1,2,3),(2,1,3),(3,1,2)$. For tripartite case, the coefficient vectors and matrices can not be put into a single matrix in a meaningful way. Nevertheless, one can still define some feature matrices as the $ {M}_{\beta\gamma}={\left[ \begin{array}{cc}
            1 & {T}_{\gamma}^{t}\\
            {T}_{\beta} & {T}_{\beta\gamma}
        \end{array}
        \right]}$ and $ {M}_{\alpha|\beta\gamma}={\left[ \begin{array}{cc}
            1 & {T}_{\beta}^{t}\otimes {T}_{\gamma}^{t}\\
            {T}_{\alpha} & {T}_{\alpha|\beta\gamma}
        \end{array}
        \right]}.$ 
% $M_{\beta^{\prime}l}={\left[ \begin{array}{cc}
% 1 & T_{l}^{t}\\
% T_{\beta^{\prime}} & T_{\beta^{\prime}l}
% \end{array}
% \right]}$ and $M_{\alpha^{\prime}|\beta^{\prime} l}={\left[ \begin{array}{cc}
% 1 & T_{\beta^{\prime}}^{t}\otimes T_{l}^{t}\\
% T_{\alpha^{\prime}} & T_{\alpha^{\prime}|\beta^{\prime} l}
% \end{array}
% \right]}$.
%which can not imply all information of quantum state individually. 
Moreover, letting $T_{\beta\gamma}$ be a column vector $vec(T_{\beta\gamma})$, which is arranged by all columns of the matrix $T_{\beta\gamma}$, one can construct ${\hat{M}}_{\alpha|\beta\gamma}={\left[ \begin{array}{cc}
            1 & {\rm vec}({T}_{\beta\gamma})^{t}\\
            {T}_{\alpha} & {T}_{\alpha|\beta\gamma}
        \end{array}
        \right]}$
% $\hat{M}_{\alpha^{\prime}|\beta^{\prime} l}={\left[ \begin{array}{cc}
% 1 & vec(T_{\beta^{\prime}l})^{t}\\
% T_{\alpha^{\prime}} & T_{\alpha^{\prime}|\beta^{\prime} l}
% \end{array}
% \right]}$ 
which has the relation ${\rm vec}({T}_{\beta\gamma}^{\prime})={O}_{\beta}^{t}\otimes {O}_{\gamma}^{t}{\rm vec}({T}_{\beta\gamma})$ 
%$vec(T_{\beta^{\prime}l}^{\prime})=O_{\beta^{\prime}}^{t}\otimes O_{l}^{t}vec(T_{\beta^{\prime}l})$
from Eq. (\ref{eq:1|23}). By performing a SVD of matrices $T_{\beta\gamma},T_{\alpha|\beta\gamma}$,  there are orthogonal matrices $P_{\beta\gamma},Q_{\beta\gamma},P_{\alpha|\beta\gamma}$, and $Q_{\alpha|\beta\gamma}$ such that %transform the above matrices into
\begin{equation}
    {T}_{\beta\gamma}={P}_{\beta\gamma}{\Sigma}_{\beta\gamma}{Q}_{\beta\gamma}^{t},{T}_{\alpha|\beta\gamma}={P}_{\alpha|\beta\gamma}{\Sigma}_{\alpha|\beta\gamma}{Q}_{\alpha|\beta\gamma}^{t},
    \label{eq:3-svd}
\end{equation}
% \begin{equation}
% T_{\beta^{\prime}l}=P_{\beta^{\prime}l}\Sigma_{\beta^{\prime}l}Q_{\beta^{\prime}l}^{t},T_{\alpha^{\prime}|\beta^{\prime} l}=P_{\alpha^{\prime}|\beta^{\prime} l}\Sigma_{\alpha^{\prime}|\beta^{\prime} l}Q_{\alpha^{\prime}|\beta^{\prime} l}^{t},
% \label{eq:3-svd}
% \end{equation}
where the singular values in $\Sigma_{\beta\gamma}$ and $\Sigma_{\alpha|\beta\gamma}$ are all in decreasing order. Similar to the two-qudit case, one can divide the singular-value matrix $\Sigma_{\beta\gamma}$ ($\Sigma_{\alpha|\beta\gamma}$) into a block matrix with $n_{\beta\gamma}$ ($n_{\alpha|\beta\gamma}$) submatrices, by putting the same singular value into a single submatrix. 
%Let the notations $n_{\beta^{\prime}l},\ n_{\alpha^{\prime}|\beta^{\prime} l}$ denote the numbers of block matrices that $\Sigma_{\beta^{\prime}l}$ and $\Sigma_{\alpha^{\prime}|\beta^{\prime} l}$ are divided into. In other word,  
Define vectors $u_{1}=P_{\beta\gamma}^{t}T_{\beta},v_{1}=Q_{\beta\gamma}^{t}T_{\gamma}, u_{2}=P_{\alpha|\beta\gamma}^{t}T_{\alpha}, v_{2}=Q_{\alpha|\beta\gamma}^{t}(T_{\beta}\otimes T_{\gamma}),u_{3}=u_{2},v_{3}=Q_{\alpha|\beta\gamma}^{t}vec(T_{\beta\gamma})$ and $n_{1}=n_{\beta\gamma},n_{2}=n_{3}=n_{\alpha|\beta\gamma}$. 
From Eqs. (\ref{eq:1|23},\ref{eq:3-svd}), one has the following theorem.
\begin{theorem}
If two three-qudit states $\rho$ and $\rho^{\prime}$ are LU equivalent, one has the following invariants: %(1)$\Sigma_{\beta^{\prime}l},\Sigma_{\alpha^{\prime}|\beta^{\prime}l}, det(T_{\beta^{\prime}l}), det(M_{\beta^{\prime}l}),$ (2)$\|\pi_{m}u_{s}\|,\|\pi_{m}v_{s}\|,m=1,2,\dots,n_{s},$ (3)$(\pi_{m}(u_{s}))^{t}(\pi_{m}(v_{s})), m=1,2,\dots,n_{s}^{\prime},$
\begin{gather}
\begin{split}
&(1)\Sigma_{\beta\gamma},\Sigma_{\alpha|\beta\gamma}, {\rm det}(T_{\beta\gamma}), {\rm det}(M_{\beta\gamma}), \notag\\
&(2)\|\pi_{m}(u_{s})\|,\|\pi_{m}(v_{s})\|,m=1,2,\dots,n_{s},\notag\\
&(3)(\pi_{m}(u_{s}))^{t}(\pi_{m}(v_{s})), m=1,2,\dots,n_{s}^{\prime},\notag
\end{split}
\end{gather}
\begin{comment}
\begin{equation}\nonumber
\begin{split}
&(1)\Sigma_{\beta^{\prime}l},\Sigma_{\alpha^{\prime}|\beta^{\prime}l}\\&(2)\|\pi_{m}(P_{\beta^{\prime}l}^{t}T_{\beta^{\prime}})\|,\|\pi_{m}(Q_{\beta^{\prime}l}^{t}T_{l})\|,m=1,2,...,n_{\beta^{\prime}l}.\\&\ \ \ \ (\pi_{m}(P_{\beta^{\prime}l}^{t}T_{\beta^{\prime}})^{t}(\pi_{m}(Q_{\beta^{\prime}l}^{t}T_{l})),m=1,2,...,n_{\beta^{\prime}l}^{\prime}.\\&(3)\|\pi_{m}(P_{\alpha^{\prime}|\beta^{\prime} l}^{t}T_{\alpha^{\prime}})\|,m=1,2,...,n_{\alpha^{\prime}|\beta^{\prime} l}.\\&\ \ \ \ \|\pi_{m}(Q_{\alpha^{\prime}|\beta^{\prime} l}^{t}(T_{\beta^{\prime}}\otimes T_{l}))\|,m=1,2,...,n_{\alpha^{\prime}|\beta^{\prime} l}.\\&\ \ \ \ (\pi_{m}(P_{\alpha^{\prime}|\beta^{\prime} l}^{t}T_{\alpha^{\prime}})^{t}(\pi_{m}(Q_{\alpha^{\prime}|\beta^{\prime} l}^{t}(T_{\beta^{\prime}}\otimes T_{l}))),m=1,2,...,n_{\alpha^{\prime}|\beta^{\prime} l}^{\prime}.\\&(4)\|\pi_{m}(P_{\alpha^{\prime}|\beta^{\prime} l}^{t}T_{\alpha^{\prime}})\|,m=1,2,...,n_{\alpha^{\prime}|\beta^{\prime} l}.\\&\ \ \ \ \|\pi_{m}(Q_{\alpha^{\prime}|\beta^{\prime} l}^{t}vec(T_{\beta^{\prime}l}))\|,m=1,2,...,n_{\alpha^{\prime}|\beta^{\prime} l}.\\&\ \ \ \ (\pi_{m}(P_{\alpha^{\prime}|\beta^{\prime} l}^{t}T_{\alpha^{\prime}})^{t}(\pi_{m}(Q_{\alpha^{\prime}|\beta^{\prime} l}^{t}vec(T_{\beta^{\prime}l}))),m=1,2,...,n_{\alpha^{\prime}|\beta^{\prime} l}^{\prime}.
\end{split}
\end{equation}
\end{comment}
where $s\in\{1,2,3\}$, $(\alpha,\beta,\gamma)=(1,2,3),(2,1,3),(3,1,2)$. In addition, $n_{1}^{\prime}$ is determined by $T_{\beta\gamma}$ and $n_{2}^{\prime}=n_{3}^{\prime}$ are determined by $T_{\alpha|\beta\gamma}$. If $T_{\beta\gamma}$ has no zero singular values, one has $n_{1}^{\prime}=n_{1}$, otherwise, $n_{1}^{\prime}=n_{1}-1$. Similar definitions are applied to $n_{2}^{\prime}$ and $n_{3}^{\prime}$.
\end{theorem}
\begin{proof}
The proof of Theorem 2 is similar to the one of Theorem 1, as the LU invariants are obtained from the feature matrices. 
\end{proof}
Theorem 2 leads to more invariants than existing results, e.g. Ref. \cite{cui2017local}. Therefore, it can be adopted to verify the LU equivalence of the states that cannot have been verified before. Moreover, the result presents a general approach to derive the necessary criteria of LU equivalence for the multipartite qudit case, based on the coefficient vectors and feature matrices constructed from the generalized Bloch representation. 
%Firstly, the generalized Bloch representation of multipartite qudit state is given according to the density matrix. Then, the corresponding coefficients can be constructed as vectors or matrices. Furthermore, one can obtain many feature matrices by combining these vectors and matrices properly. Finally, necessary criteria of LU equivalence for multipartite quantum states can be derived from these feature matrices. 
Although, for simplicity, we discussed the qudit case with the same dimensions, cases with different dimensions can be studied similarly. Therefore, this method can always be used to derive the necessary conditions for the LU equivalence of arbitrary-dimensional multipartite mixed quantum states. 

\section{Invariants and entanglement}\label{sec:entanglement}
Entanglement is one of the most important quantum resources, which is generally applied into various quantum information processing tasks. However, the complete characterization of entanglement is extremely challenging. It is known that a complete set of invariants under LU transformation can be adopted to depict the properties of entanglement. For pure states, the constructed invariants can be used to derive the concurrence. Here, we show how our invariants lead to the estimation of concurrence.%, which is an entanglement measure \cite{PhysRevA.64.042315,PhysRevLett.95.260502,PhysRevLett.92.167902,PhysRevA.62.032307,PhysRevLett.80.2245,doi:10.1080/09500340210121589}.

For simplicity, let us consider a two-qudit pure state $\rho$. It is known that the concurrence of a pure bipartite state is given by $C_{N}(\rho)=\sqrt{(N/(N-1))(1-Tr(\rho_{1}^{2}))}$ \cite{con2002,PhysRevA.88.062116}, where $\rho_{1}$ is the partial trace of state $\rho$ by tracing out the subsystem ${\cal H}_{2}$. By combining the Eq.$~$(\ref{eq:rho}), one has $1-C_{N}^{2}(\rho)=\frac{2N^{3}}{N-1}\|R\|^{2}$. From the proof of Theorem 1, one can obtain that $\|R\|=\|\widetilde{R}\|=\sqrt{\sum_{m}\|\pi_{m}(\widetilde{R})\|^{2}}$. As $\rho$ is a pure state, we have $\|R\|=\|S\|=\|\widetilde{S}\|=\sqrt{\sum_{m}\|\pi_{m}(\widetilde{S})\|^{2}}$. Therefore, one has 
\begin{equation}\nonumber
\sum_{m}\|\pi_{m}(\widetilde{R})\|^{2}=\sum_{m}\|\pi_{m}(\widetilde{S})\|^{2}=\frac{N-1}{2N^{3}}(1-C_{N}^{2}(\rho)). 
\end{equation}
By adopting the invariants $\|\pi_{m}(\widetilde{R})\|,\|\pi_{m}(\widetilde{S})\|$, one can measure the entanglement of any pure two-qudit states. %If $\sum_{m}\|\pi_{m}(\widetilde{R})\|^{2}\ \rm{or}\ \sum_{m}\|\pi_{m}(\widetilde{S})\|^{2}<(N-1)/2N^{3}$, the state is entangled. Thus, the constructed invariants can show an excellent performance for quantifying entanglement for pure bipartite states. In other word, 
Therefore, the result simply reveals that the concurrence can be expressed by our LU invariants in the scenario of bipartite pure states. The close relation between LU invariants and entanglement is naturally constructed. Such an intimately deep connection will promote the research of entanglement measure for mixed states based on LU invariants.

\section{Conclusion}\label{sec:conclusion}
In this paper, we have presented the necessary criteria for the local unitary equivalence of arbitrary dimensional multipartite mixed quantum states. By using the criteria, one can verify the local unitary equivalence of states with density matrices even having degenerate eigenvalues. 
%The method is mathematically operational in the sense that the criteria are given by some local unitary invariants. 
It can be easily extended to the multipartite case, which gives rise to the necessary criteria for the local unitary equivalence of multipartite states based on feature matrices. In addition, the relation between the local unitary invariants and entanglement, is shown with one example of pure bipartite states. 
%For pure two-qubit states, a quadratic relation between the local unitary invariants and the concurrence has been analytically obtained. Furthermore, by numerical calculation, similar relation has been observed for any pure two-qudit states. 
In a similar way, one can find other relations between local unitary invariants and entanglement measure for any pure multipartite states. 
%though, the computational complexity will increase fast with the dimension and parties. 
We hope that a full characterization of the entanglement properties for mixed states can be motivated from the constructed invariants.

In the future, constructing more local unitary invariants to obtain sufficient and necessary conditions for verifying the local unitary equivalence of arbitrary-dimensional multipartite mixed quantum states is interesting and challenging. Meanwhile, more local unitary invariants can promote the exploration for characterizing entanglement properties. 
%As the group $SO(N^{2}-1)$ is not isomorphic to $SU(N)$ any more when $N>2$, our derived necessary conditions can not be sufficient. In the future work, we hope to find more local unitary invariants to guarantee that the derived operators $U_{1},U_{2}$ in Eq.$~$(\ref{eq:LU}) from invariants are actually in group $SU(N)$. And then, one can obtain sufficient and necessary conditions for verifying local unitary equivalence of arbitrary dimensional multipartite mixed quantum states. Meanwhile, more local unitary invariants can promote the exploration for characterizing entanglement properties. 

%\ack
\section*{Acknowledgments}

This work has been supported by the National Natural Science Foundation of China (Grants No. 62375252, No. 62031024, No. 11874346, No. 12174375, No. 12005091, No. 12075159, No. 12171044), the National Key R$\& $D Program of China (2019YFA0308700), the Anhui Initiative in Quantum Information Technologies (AHY060200), the Innovation Program for Quantum Science and Technology (No. 2021ZD0301100), the Academician Innovation Platform of Hainan Province, the Research Startup Foundation of Hangzhou Dianzi University (No. KYS275623071) and Zhejiang Provincial Natural Science Foundation of China under Grant No. LQ24A050005.

\section*{Appendix A: details for matrices $T_{2|13},T_{3|12}$}\label{sec:appen2}
\setcounter{equation}{0}
\renewcommand\theequation{A\arabic{equation}}
Here, the detailed forms of matrices $T_{2|13},T_{3|12}$ are presented as following
% \begin{equation}
% T_{1|23}=
% \begin{bmatrix}
% T_{111} & T_{112}& \cdots & T_{1(N^{2}-1)(N^{2}-1)}\\
% T_{211} &T_{212}& \cdots & T_{2(N^{2}-1)(N^{2}-1)}\\
% \vdots & \vdots & \ddots & \vdots\\
% T_{(N^{2}-1)11} & T_{(N^{2}-1)12}& \cdots & T_{(N^{2}-1)(N^{2}-1)(N^{2}-1)}
% \end{bmatrix},
% \end{equation}
\begin{equation}
T_{2|13}=
\begin{bmatrix}
T_{111} & T_{112}& \cdots & T_{(N^{2}-1)1(N^{2}-1)}\\
T_{121} &T_{122}& \cdots & T_{(N^{2}-1)2(N^{2}-1)}\\
\vdots & \vdots & \ddots & \vdots\\
T_{1(N^{2}-1)1} & T_{1(N^{2}-1)2}& \cdots & T_{(N^{2}-1)(N^{2}-1)(N^{2}-1)}
\end{bmatrix},
\end{equation}
\begin{equation}
T_{3|12}=
\begin{bmatrix}
T_{111} & T_{121}& \cdots & T_{(N^{2}-1)(N^{2}-1)1}\\
T_{112} &T_{122}& \cdots & T_{(N^{2}-1)(N^{2}-1)2}\\
\vdots & \vdots & \ddots & \vdots\\
T_{11(N^{2}-1)} & T_{12(N^{2}-1)}& \cdots & T_{(N^{2}-1)(N^{2}-1)(N^{2}-1)}
\end{bmatrix}.
\end{equation}
For any multipartite case, the coefficient matrix $T_{\alpha|\beta\gamma\dots}$ can be arranged in a similar way, where the $\alpha$-th subscript is used as the row index.
%\bibliographystyle{apsrev4-2}
%\bibliography{bib/lu}
%\input{main.bbl}


\end{document}